\documentclass[11pt]{article}

\usepackage{times}

\usepackage[T1]{fontenc}
\usepackage{epsfig}
\usepackage{rotating}
\usepackage{graphicx}
\usepackage{float}
\usepackage{booktabs}

\usepackage{amsmath}
\usepackage{amssymb}
\usepackage{fullpage}
\usepackage{amsthm}
\usepackage{graphicx}
\usepackage{color}
\usepackage{enumerate}

\newcommand{\ignore}[1]{}

\renewcommand{\l}{\ell}

\newtheorem{theorem}{Theorem}

\newtheorem{prop}[theorem]{Proposition}

\newtheorem{corollary}[theorem]{Corollary}
\newtheorem{definition}{Definition}

\newtheorem{remark}{Remark}

\renewcommand{\le}{\leqslant}
\renewcommand{\ge}{\geqslant}

\newcommand{\eps}{\varepsilon}
\renewcommand{\epsilon}{\varepsilon}

\newcommand{\vnote}[1]{}
\newcommand{\anote}[1]{}

\newcommand{\F}{\mathbb{F}}
\newcommand{\R}{\mathbb{R}}

\renewcommand{\H}{\mathcal{H}}

\newcommand{\A}{{\mathcal{A}}}

\newcommand{\au}{almost universal}
\newcommand{\se}{storage enforceable}

\title{{\bf Almost Universal Hash Families are also Storage Enforcing}}

\author{Mohammad Iftekhar Husain\and Steve Ko\and Atri Rudra\footnote{Supported in part by NSF CAREER grant CCF-0844796}\and Steve Uurtamo\footnotemark[1]}

\date{Department of Computer Science and Engineering,\\
University at Buffalo, SUNY,\\
Buffalo, NY, 14214.\\
{\{imhusain,stevko,atri,uurtamo\}@buffalo.edu}}

\begin{document}
\maketitle

\setcounter{page}{0}
\thispagestyle{empty}

\begin{abstract}
We show that every almost universal hash function also has the \textit{storage enforcement} property. Almost universal hash functions have found numerous applications and we show that this new storage enforcement property allows the application of almost universal hash functions in a wide range of remote verification tasks: (i) Proof of Secure Erasure (where we want to remotely erase and securely update the code of a compromised machine with memory-bounded adversary), (ii) Proof of Ownership (where a storage server wants to check if a client has the data it claims to have before giving access to deduplicated data) and (iii) Data possession (where the client wants to verify whether the remote storage server is storing its data). Specifically, storage enforcement guarantee in the classical data possession problem removes any practical incentive for the storage server to cheat the client by saving on storage space.

The proof of our result relies on a natural combination of Kolmogorov Complexity and List Decoding. To the best of our knowledge this is the first work that combines these two techniques. We believe the newly introduced storage enforcement property of almost universal hash functions will open promising avenues of exciting research under memory-bounded (bounded storage) adversary model.
\end{abstract}


\newpage

\section{Introduction}

Universal hashing was defined in the seminal paper of Carter and Wegman~\cite{CW79} and since then has been an integral part of both complexity and algorithms research. In this paper we consider the following well known generalization of universal hashing called $\eps$-\au~\cite{au}. In particular, a family $\mathcal H=\{H_1,\dots,H_n\}$ of hash functions (where $H_i:[q]^k \rightarrow [q]$)\footnote{Our results can also handle the case where each each $H_i:[q]^k\rightarrow [q_i]$ for potentially different $q_i$ for each $i\in [n]$. For simplicity, we will concentrate on the case of $q_i=q$ for every $i\in [n]$.} is an $\eps$-\au\ hash family if for every $x\neq y\in [q]^k$, we have
\[\Pr_{i\in [n]}\left[ h_i(x)=h_i(y)\right]\le \eps.\]
($1/q$-\au\ hash families are universal hash families as defined in~\cite{CW79}.)

In this paper, we show that any $\eps$-\au\ hash family also has the ``storage enforcement" property. 
In particular, we call a hash function family $\mathcal H=\{H_1,\dots,H_n\}$ to be $(\gamma,f(\cdot))$-\se\ if the following holds. A \textit{prover} claims to have $x\in [q]^k$. The prover is allowed to perform arbitrary computation on $x$ and retain its output $y\in [q]^*$ with it. Then the following holds for any $x\in [q]^k$: if the prover can compute\footnote{We will assume the prover uses an algorithm $\A_x$ to compute $\A_x(i,y)$ as its version of $h_i(x)$.} $h_i(x)$ with probability at least $\gamma$ (for a uniformly random $i\in [n]$), then $|y|\ge f(x)$.

Our main result is the following:

\begin{theorem}
\label{thm:au-se}
If a computable hash function family $\mathcal H=\{H_1,\dots,H_n\}$ is $\eps$-\au, then it is also $(\sqrt{\eps},f(\cdot))$-\se, where $f(x)\approx C(x)$, where $C(x)$ is the plain Kolmogorov complexity of $x$.
\end{theorem}

The storage enforcement property is interesting in its own right and we present its applications in problems of proof of secure erasure, proof of ownership and data possession.
 
Before discussing the applications of storage enforcement, we present the
techniques we use to prove Theorem~\ref{thm:au-se}, which to the best of our
knowledge is the first instance of combining list
decoding~\cite{List-Decoding-Orig-Elias, List-Decoding-Orig-Wozencraft} and
Kolmogorov complexity~\cite{Kolmogorov-Complexity-Orig,
Kolmogorov-Complexity-Recent} naturally to give interesting results. Both
Kolmogorov complexity (see the textbook~\cite{Kolmogorov-Complexity-Recent})
and list decoding (see the survey by Sudan~\cite{sudan-sigact} and Guruswami's
thesis~\cite{G-thesis}) have found numerous applications in complexity theory.
We believe that this combination merits more exploration and hope that this work
will lead to a more systematic study.

\paragraph{Our techniques.} 



We begin our discussion with the following related result. Let $H:[q]^k\rightarrow [q]^n$ be an error-correcting code
with good \textit{list decodability} properties-- see Section~\ref{sec:coding} for coding definitions. (We do not need an algorithmic guarantee, as just a combinatorial guarantee suffices.) Consider the natural hash family $\mathcal H=\{h_1,\dots,h_n\}$, where $h_i(x)=H(x)_i$. We now argue that $\mathcal H$ has storage enforcement property.
In particular, we quickly sketch why the prover cannot get away with storing a vector $y$ such that $|y|$ is smaller than $C(x)$, the plain Kolmogorov complexity of $x$, by some appropriately small additive factor. Since we are assuming that the prover uses an algorithm $\A_x$ to compute its answer $\A_x(\beta,y)$ to the random challenge $\beta\in [n]$, if the prover's answer is correct with probability at least $\gamma$, then note that the vector $(\A_x(\beta,y))_{\beta\in [n]}$ differs from $H(x)$ in at most $1-\gamma$ fraction of positions. Thus, if $H$ has good list decodability, then (using $\A_x$) one can compute a list $\{x_1,\dots,x_L\}$ that contains $x$. Finally, one can use $\log{L}$ bits of advice (in addition to $y$) to output $x$. This procedure then becomes a description for $x$ and if $|y|$ is sufficiently smaller than $C(x)$, then our description will have size $<C(x)$, which contradicts the definition of $C(x)$.

The above result then implies Theorem~\ref{thm:au-se} because (1) Every $\eps$-\au\ hash family defines a code $H$ (in exactly the same way as we defined corresponding hash family $\H$ from the code $H$ in the paragraph above) with (relative) distance at least $1-\eps$; (2) Every code $H$ with relative distance $1-\eps$ can be (combinatorially) list decoded from $1-\sqrt{\eps}$ fraction of errors by the Johnson bound (cf.~\cite{G-thesis}).


From a more practical point of view, 
the Karp Rabin hash corresponds to $H$ being so-called the ``Chinese Remainder
Theorem" (CRT) code and the polynomial hash corresponds to $H$ being the
Reed-Solomon code. Reed-Solomon and CRT codes have good list decodability, which
implies that our protocols can be implemented using the classic Karp Rabin and
polynomial hashes.

%
%


Next, we present three applications of hash functions that have storage enforcement.

\paragraph{Proof of Secure Erasure.} 
%
Often, it is important to verify the internal state of a remote embedded device
such as a sensor, actuator or computer peripheral with limited memory. This is
to assure that it is running the intended code and not a malicious or arbitrary
code. When any suspicious behavior of the remote device is detected, sometimes
it becomes necessary to erase the entire memory content and then update it with
legitimate code. Two crucial components of this process is to make sure that (i)
the remote device has to store a  random string at least as large as the memory
size and (ii) efficiently verify that the device has indeed stored the random
string  remotely. A similar problem has been studied before, under the name of
\textit{Proof of Secure Erasure}~\cite{crypto-storage-enforcement} using
cryptographic primitives which is, in most cases, too expensive for
resource-constrained embedded devices. However, almost universal hash functions
with the storage enforcement property applies naturally to provide an efficient
proof of secure erasure scheme as outlined below.

The prover can use any \se\ hash family, along with a randomly-generated string $x$, (which will, with constant probability greater than $1/2$, have $C(x)\ge (|x|-O(1))$)\footnote{The exact bound is very strong. $C(x)\ge (|x|-r)$ for at least a $(1-1/{2^r})$-fraction of the strings of length $|x|$.}, to force the remote device (verifier) to first store a string of length close to $|x|$ (overwriting any malicious code as long as $|x|$ is chosen large enough), second, to correctly answer a verification request (a failure to answer the request proves that the initial erasure has failed), and third (now that there is no room for a malicious code) to install the updated code.

Apart from efficiency issues, the cited work~\cite{crypto-storage-enforcement} also assumes that the remote device stores some portion of its trusted code in a fixed ROM. However, a universal hash function based scheme with storage enforcement is not strictly dependent on such an assumption. We do, however, (as does \cite{crypto-storage-enforcement}) require that the remote machine not have network access to an external storage  which is a very reasonable assumption in memory bounded adversary model. 


\paragraph{Proof of Ownership.} Consider the case when a client wants to upload it's data
$x$ to a storage service (such as Dropbox) that uses `deduplication.' In
particular, if $x$ is already stored on the server, then the server will ask the
client not to upload $x$ and give it ownership of $x$. To save on communication,
the sever asks the client to send it a hash $h(x)$ and if it matches the hash of the stored $x$ on the server, the client is issued an ownership of $x$. As identified by \cite{Halevi-Ownership}, this simple deduplication scheme can be abused for malicious purposes. For example,
if a malicious user gets hold of the hashes of the files stored on the server (through server break-in or unintended leakage), it will get access to those files. The simple deduplication scheme can also be exploited as a unintended content distribution network. A malicious user can upload a potentially large file and share the hash of the file with accomplices. Now, all these accomplices can present the hash to storage server and get access to large file as if it is a content distribution network. A \se\  hash function can address such a situation. Specifically, if one has a hash family $\{h_1,\dots,h_n\}$ that is \se, then the server can pick a random $i\in [n]$ and ask the client to send it $h_i(x)$. If the client succeeds in doing so then, the server knows the client has a lot of information (close to $C(x)$) related to $x$ and not just a small fingerprint of $x$.

Proof of ownership is studied formally in \cite{Halevi-Ownership}. In addition to the scheme being computation intensive, the security guarantee in \cite{Halevi-Ownership} is based on the min-entropy of the original file dealing with the probability distribution of input data blocks, which we believe (as also agreed by the authors of \cite{Halevi-Ownership}) is not the best measure to enforce storage. We do not make any assumptions on what kind of data we are storing. Due to this, Kolmogorov complexity is a better option as it is defined for each specific block of data, which is at the core of our solution. Also, majority of the existing data possession schemes (as discussed in the next application) could not be used to address proof of ownership. Because, they require the prover (client) to store $x$ in a modified form which is an infeasible assumption for this application.

\paragraph{Data Possession.} This application considers the following
problem: A string $x\in [q]^k$ is held by the client/user, then transmitted to a
remote server (or broken up into several pieces and transmitted to several
remote servers). At some later point, after having sent $x$, the client would
like to know that the remote server(s) is storing $x$ correctly. 
The problem has been studied under the
name of \textit{proof of retrievability} (cf.~\cite{JuelsPOR,BowersPOR,DVW09})
and/or \textit{data possession}
(cf.~\cite{AteniesePDP,AtenieseSEPDP,CurtmolaMRPDP,ErwayDPDP}). 
Given the greater prevalence of outsourcing of storage (e.g. in ``cloud computing''),
such a verification procedure is crucial for auditing purposes to make sure that
the remote servers are following the terms of their agreement with the client.

In contrast with the existing schemes, our approach to the remote server adversarial model is rather practical. We assume that the storage provider is not fully malicious: e.g., it is reasonable to assume that Amazon will not try and mess with an individual's storage
in a malicious way. However, the server will not be completely honest either. We
make the natural assumption that the servers' main goal is to save on its storage
space (since it practically helps the server to save on the cost of storage medium, networking and power): so if it can pass the verification protocol by storing substantially less than $|x|$ amounts of data, then it will do so.\footnote{In other words, the server is
willing to spend a lot of computational power to save on its space usage but is
not maliciously trying to cheat the client.} In our scheme, we will strive to
make this impossible for the server to do without having a high probability of
getting caught. 

Also, before we move on to our scheme, we would like to make a point that does not seem to have been made explicitly before. Note that we cannot prevent a server from compressing their copy of the string, or otherwise representing it in any reversible way. (Indeed, the user would not care as long as the server is able to recreate $x$.) This means that a natural upper bound on the amount of storage we can force a server into is $C(x)$, the (plain) Kolmogorov complexity of $x$, which is the size of the smallest (algorithmic) description of $x$.

Our scheme to address the data possession problem is determined by a $(\gamma,f(\cdot))$-\se\ hash family $\{h_1,\dots, h_n\}$. In the pre-processing step, the client picks a random $i\in [n]$, stores $(i,h_i(x))$ and ships off $x$ to the server. During the verification step, the client send $i$ and asks for $h_i(x)$. Upon receiving the reply $z$ from the server, the client checks whether $z=h_i(x)$.
Since the hash family was \se,  if the server passes the verification protocol with probability $\gamma>0$, then they provably have to store $f(x)$ (which in our results will be $C(x)$ up to a very small \textit{additive} factor) bits of data. A cheating server is allowed to use arbitrarily powerful algorithms (as long as they terminate) while responding to challenges from the user.\footnote{The server(s) can use different algorithms for different strings $x$ but the algorithm cannot change across different challenges for the same $x$.} Further, every honest server only needs to store $x$. In other words, unlike some existing results, our protocol does not require the server(s) to store a massaged version of $x$. In practice, this is important as we do not want our auditing protocol to interfere with ``normal" read operations. However, unlike many of the existing results based upon cryptographic primitives, our protocol can only allow a number of audits that is proportional to the user's local storage.\footnote{An advantage of proving security under cryptographic assumptions is that one can leverage existing theorems to prove additional security properties of the verification protocol. We do not claim any additional security guarantees other than the ability of being able to force servers to store close to $C(x)$ amounts of data.} Further,  as a demonstration of the versatility of our technique, it'll turn out that our schemes can also provide a proof of retrievability (see e.g.
Remark~\ref{rem:por}), but we do not claim any improvement over existing work on proofs of retrievability, e.g.~\cite{DVW09}.
In general, our single prover scheme is very similar to one of the schemes discussed in ~\cite{DVW09}. For example, one of the instantiations of our scheme with Reed-Solomon codes in fact resembles exactly one of their scheme. However, some crucial differences in technical contribution of the schemes are the following. First, even though both schemes use list decodability of codes, the schemes differ in how one ``prunes" the output of list decoding. The scheme in~\cite{DVW09} uses extra hash functions while our work uses Kolmogorov complexity. Second and perhaps more importantly, we believe that the main contribution of our work is to formally establish the storage enforcement properties for every almost universal hash function, and showing a range of applications of these \se\  hash functions.

Our techniques are very flexible. As mentioned earlier, our protocols also imply a proof of retrievability. As a further demonstration of this flexibility, we show that our results seamlessly generalize to the multiple server case. We show that we can allow for arbitrary collusion among the servers if we are only interested in checking if at least one sever is cheating (but we cannot identify it). If we want to identify at least one cheating server (and allow for servers to not respond to challenges), then we can handle a moderately limited collusion among servers. 

Theorem~\ref{thm:au-se} along with the above then implies that any \au\ hash family implies a valid verification protocol for the data possession problem. Our protocol with best parameters, which might need the user and honest server(s) to be exponential time algorithms, provably achieves the optimal local storage for the user and the optimal total communication between the user and the server(s). 
By picking an appropriate \au\ hash family, with slightly non-optimal storage and communication parameters, the user and the honest server(s) can work with single pass logarithmic space data stream algorithms.

\section{Preliminaries}
\label{sec:prelim}

We now formally define different parameters of a verification protocol based on \se\   hash functions. For the secure erasure application, the administrator/user is the verifier and the remote embedded device is the prover. In the proof of ownership problem, the deduplication server is the verifier and the client is the prover. For the data possession problem, the client is the verifier and the storage server is the prover. For clarity, rest of the discussion will focus on the data possession application to accomodate both single prover and multiple prover cases.

\paragraph{Basic Verification Scheme.}

We use $\mathbb{F}_{q}$ to denote the
finite field over $q$ elements. We also use $[n]$ to denote the set
$\{1,2,..,n\}$. Given any string $x\in [q]^*$, we use $|x|$ to denote the length of $x$ in bits. Additionally, all logarithms will be base 2 unless otherwise specified. We use $U$ to denote the \textit{verifier}. 

For the data possession problem, we define the verification protocol in terms of multiple prover case first and then discuss single prover as a specific case for ease of dicussion. We assume that $U$ wants to store its data $x\in [q]^k$ among $s$ service providers (or provers) $P_1,\dots,P_s$. For the proof of ownership problem, the client (or the prover) $P_1$ has the data $x$. In the preprocessing step, the provers $\mathcal P=\{P_1,\dots,P_s\}$ get $x$ divided up equally among the $s$ provers -- we will denote the chunk for prover $i\in [s]$ as $x_i\in [q]^{n/s}$.\footnote{We will assume that $s$ divides $n$. In our results for the case when $H$ is a linear code, we do not need the $x_i$'s to have the same size, only that $x$ can be partitioned into $x_1,\dots,x_s$. We will ignore this possibility for the rest of the paper.} (For the data possession problem, we can assume that $U$ does this division and sends $x_i$ to $P_i$.) Each prover is then allowed to apply any computable function to its chunk and to store a string $y_i\in [q]^*$. Ideally, we would like $y_i=x_i$. However, since the provers can compress $x$, we would at the very least like to force $|y_i|$ to be as close to $C(x_i)$ as possible. For notational convenience, for any subset $T\subseteq [s]$, we denote $y_T$ ($x_T$ resp.) to be the concatenation of the strings $\{ y_i\}_i\in T$ ($\{x_i\}_{i\in T}$ resp.).

To enforce the conditions above, we design a protocol. We will be primarily concerned with the amount of storage at the verifier side and the amount of communication and want to minimize both simultaneously while giving good verification properties. The following definition captures these notions. (We also allow for the provers to collude among each other.)

   \begin{definition}
Let $s,c,m\ge 1$ and $0\le r\le s$ be integers, $0\le \rho\le 1$ be a real and $f:[q]^*\rightarrow \R_{\ge 0}$ be a function. Then an $(s,r)$-{\em party verification protocol} with resource bound $(c,m)$ and verification guarantee $(\rho,f)$ is a randomized protocol with the following guarantee.
For any string $x\in [q]^k$, $U$ stores at most $m$ bits and communicates at most $c$ bits with the $s$ provers. At the end, the protocol either outputs a $1$ or a $0$. Finally, the following is true for any $T\subseteq [s]$ with $|T|\le r$: If the protocol outputs a $1$ with probability at least $\rho$, then assuming that every prover $i\in [s]\setminus T$ followed the protocol and that every prover in $T$ possibly colluded with one another, we have $|y_T|\ge f(x_T)$.
\end{definition}

We will denote a $(1,1)$-party verification protocol as a one-party verification protocol. (Note that in this case, the single prover is allowed to behave arbitrarily.)

All of our protocols will have the following structure: we first pick a hash family. The protocol will pick random hash(es) and store the corresponding hash values for $x$ (along with the indices of the hash functions) during the pre-processing step. During the verification step, $U$ sends the indices to the hashes as challenges to the $s$ provers. Throughout this paper, we will assume that each prover $i$ has a computable algorithm $\A_{x,i}$ such that on challenge $\beta$ it returns an answer $\A_{x,i}(\beta,y_i)$ to $U$. The protocol then outputs $1$ or $0$ by applying a (simple) boolean function on the answers and the stored hash values.

In particular, for the one-party verification protocol, we have the following simple result, which we prove in the appendix.

\begin{prop}
\label{prop:au-verification}
Let $\H=\{h_1,\dots,h_n\}$ be a hash family where $h_i:[q]^k\rightarrow [q]$ that is $(\gamma,f)$-\se. Then there exists a one party verification protocol with resource bound $(\log{q}+\log{n},\log{q}+\log{n})$ and performance guarantee $(\gamma,f)$.
\end{prop}

\subsection{Coding Basics}
\label{sec:coding}

We begin with some basic coding definitions. 
An error-correcting \textit{code} $H$ with \textit{dimension} $k\ge 1$ and \textit{block length} $n\ge k$ over an \textit{alphabet} of size $q$ is any function $H: [q]^k \rightarrow [q]^n$.
   A \textit{linear code} $H$ is any error-correcting code that is a linear function, in which case we correspond $[q]$ with $\F_q$.
   A \textit{ message} of a code $H$ is any element in the domain of $H$.
   A \textit{ codeword} in a code $H$ is any element in the range of $H$. 
The \textit{ Hamming distance} $\Delta(x,y)$ of two same-length strings is the number of symbols in which they differ. The \textit{ relative distance} $\delta$ of a code is $\min_{x\neq y}\frac{\Delta(x,y)}{n}$, where $x$ and $y$ are any two different codewords in the code.

The following connection between \au\ and codes is well known. The straightforward proof is in the appendix.

\begin{prop}[\cite{au-code}]
\label{prop:au-code}
Let $\H=\{h_1,\dots,h_n\}$ be an $\eps$-\au\ hash family with $h_i:[q]^k\rightarrow [q]$. Now consider the related code $q$-ary code $H$ with block length $n$ and dimension $k$ where $H(x)=(h_i(x))_{i\in [n]}$. Then $H$ has relative distance at least $1-\eps$.
\end{prop}

   \begin{definition} A $(\rho,L)$ {\em list-decodable code} is any error-correcting code such that for every vector $e$ in the codomain of $H$, the number of codewords that are Hamming distance $\rho n$ or less from $e$ is always $L$ or fewer.
\end{definition}


\paragraph{Johnson Bound.}
We now state a general combinatorial result for list decoding codes with large distance, which will be useful. The result below allows for a sort of non-standard definition of codes, where a codeword is a vector in $\prod_{i=1}^n [q_i]$, where the $q_i$'s can be distinct.\footnote{We're overloading the product operator $\prod$ here to mean the iterated Cartesian product.} (So far we have looked only at the case where $q_i=q$ for $i\in [n]$.) The notion of Hamming distance still remains the same, i.e. the number of positions that two vectors differ in. (The syntactic definitions of the distance of a code and the list decodability of a code remain the same.) We will need the following result:

\begin{theorem}[\cite{G-thesis}]
\label{thm:johnson}
Let $C$ be a code with block length $n$ and distance $d$ where the $i$th symbol in a codeword comes from $[q_i]$. Then the code is $\left(1-\sqrt{1-\frac{d}{n}},2\sum_{i=1}^n q_i\right)$-list decodable.
\end{theorem}

Theorem~\ref{thm:johnson} and Proposition~\ref{prop:au-code} implies the following:
\begin{corollary}
\label{cor:au-ld}
Let $\H=\{h_1,\dots,h_n\}$ be an $\eps$-\au\ hash family with $h_i:[q]^k\rightarrow [q]$. Now consider the related code $q$-ary code $H$ where $H(x)=(h_i(x))_{i\in [n]}$. Then $H$ is $(1-\sqrt{\eps},2qn)$-list decodable.
\end{corollary}

   For the purposes of this paper, we only consider codes $H$ that are members of a family of codes $\bar{\mathcal{H}}$, any one of which can be indexed by $(k,n,q,\rho n)$. 

   \paragraph{Plain Kolmogorov Complexity.}

   \begin{definition}
   The {\em plain Kolmogorov Complexity} $C(x)$ of a string $x$ is the minimum sum of sizes of a compressed representation of $x$, along with its decompression algorithm $D$, and a reference universal Turing machine that runs the decompression algorithm.
   \end{definition}

   Because the reference universal Turing machine size is constant, it is useful to think of $C(x)$ as simply measuring the amount of inherent (i.e. incompressible) information in a string $x$.


\section{Single Prover Storage Enforcement}
\label{one-party-result}

We begin by presenting our main result for the case of one prover ($s=1$) to illustrate the combination of list decoding and Kolmogorov complexity. In the subsequent section, we will generalize our result to the multiple prover case (for the data possession problem).

\begin{theorem}
\label{thm:ld-se}
For every computable error-correcting code $H: [q]^k \rightarrow [q]^n$ that is $(\rho,L)$ list-decodable, define the hash family $\H=\{h_1,\dots,h_n\}$ as $h_i(x)=H(x)_i$. Then $\H$ is a $(1-\rho,f)$-\se, where  for any $x\in [q]^k$, $f(x)=C(x) - \log (qLn^3) - 2\log\log (qn) -c_0$, for some fixed constant $c_0$.\footnote{The contribution from the encoding of the constants in this theorem to $c_0$ is 2. For most codes, we can take $c_0$ to be less than a few thousand. What is important is that the contribution from the encoding is independent of the rest of the constants in the theorem. Although we do not explicitly say so in the body of the theorem statements, this important fact is true for the rest of the results in this paper as well.}
\end{theorem}

\begin{proof} We assume that the prover, upon receiving $x$, saves a string $y\in [q]^*$. The prover is allowed to use any computable function to obtain $y$ from $x$.
Further, we assume that the prover, upon receiving the (random) challenge $\beta\in [n]$, uses a computable function $\A_x:[n]\times [q]^*\rightarrow [q]$ to compute $a=\A_x(\beta,y)$.

We prove the \se\ property by contradiction. Assume that $|y|<f(x)\stackrel{def}{=} C(x) - \log (qLn^3) - 2\log\log (qn) -c_0$ and yet $\A_x(\beta,y)=h_{\beta}(x)=H(x)_{\beta}$ with probability at least $1-\rho$ (over the choice of $\beta$). Define $z=(\A_x(\beta,y))_{\beta\in [n]}$. Note that by the claim on the probability,  $\Delta(z,H(x))\le \rho n$. We will use this and the list decodability of the code $H$ to prove that there is an algorithm with description size $< C(x)$ to describe $x$, which is a contradiction. To see this, consider the following algorithm that uses $y$ and an advice string $v\in\{0,1\}^{|{L}|}$:
\begin{enumerate}
\item Compute a description of $H$ from $n,k,\rho n$ and $q$.
\item Compute $z=(\A_x(\beta,y))_{\beta\in [n]}$.
\item By cycling through all $x\in [q]^k$, retain the set $\mathcal L\subseteq [q]^k$ such that for every $u\in \mathcal L$, $\Delta(H(u),z)\le \rho n$.
\item Output the $v$th string from $\mathcal L$.
\end{enumerate}

Note that since $H$ is $(\rho,L)$-list decodable, there exists an advice string $v$ such that the algorithm above outputs $x$. Further, since $H$ is computable, there is an algorithm $\mathcal E$ that can compute a description of $H$ from $n,k,\rho n$ and $q$. (Note that using this description, we can generate any codeword $H(u)$ in step 3.) Thus, we have description of $x$ of size $|y|+|v|+|\A_x|+|\mathcal E|+ (3\log{n}+\log{q}+2\log\log n+2\log\log q+2)$ (where the last term is for encoding the different parameters\footnote{We use a simple self-delimiting encoding of $q$ and $n$, followed immediately by $k$ and $\rho n$ in binary, with the remaining bits used for $v$. A simple self-delimiting encoding for a positive integer $u$ is the concatenation of: $( \lceil\log(|u|)\rceil$ in unary, $0$, $|u|$ in binary, $u$ in binary). We omit the description of this encoding in later proofs.}),
which means if $|y| < C(x)-|v|-|\A_x|-|\mathcal E|-(3\log{n}+\log{q}+2\log\log n+2\log\log q+2) = f(x)$, then we have a description of $x$ of size $<C(x)$, which is a contradiction.

\end{proof}

The one unsatisfactory aspect of the result above is that if $H$ is not polynomial time computable, then Step 2 in the pre-processing step for $U$ is not efficient. Similarly, if the sever is not cheating (and e.g. stores $y=x$), then it cannot also compute the correct answer efficiently. We will come back to these issues when we instantiate $H$ by an explicit code such as Reed-Solomon.

\begin{remark}
\label{rem:por}
If $H$ has relative distance $\delta$, then note that $\H$ is $(1-\delta/2-\epsilon,C(x)-\log(qn^3)-2\log\log(qn)-c_0)$-\se\ for some fixed constant $c_0$. Further, $y$ has enough information for the prover to compute $x$ back from it. (It can use the same algorithm to compute $x$ from $y$ detailed above, except it does not need the advice string $v$, as in Step 3 we will have $\mathcal L=\{x\}$.) For the more general case when $H$ is $(\rho,L)$-list decodable and $\H$ is $(1-\rho,C(x)-\log{qLn^3}-2\log\log{qn}-c_0)$-\se, then $y$ has enough information for the prover to compute a list $\mathcal{L}\supseteq \{x\}$ with $|\mathcal L|\le L$. The verifier, if given access to $\mathcal L$, can use its randomly stored $h_{\beta}(x)$ to pick $x$ out of $\mathcal L$ with probability at least $1-\delta L$.
\end{remark}

Theorem~\ref{thm:ld-se} along with Proposition~\ref{prop:au-verification} implies the following result:
\begin{theorem}
\label{thm:main-single}
For every computable error-correcting code $H: [q]^k \rightarrow [q]^n$ that is $(\rho,L)$ list-decodable, there exists a one-party verification protocol with resource bound $(\log{n} + \log{q},\log{n} + \log{q})$ and verification guarantee $(1-\rho,f)$, where  for any $x\in [q]^k$, $f(x)=C(x) - \log (qLn^3) - 2\log\log (qn) -c_0$, for some fixed constant $c_0$.
\end{theorem}

\section{Multiple Prover Storage Enforcement}

In this section, we show how to extend the result from Theorem~\ref{thm:main-single} to the multiple server case. In the first two sub-sections, we will implicitly assume the following: (i) We are primarily interested in whether some server was cheating and not in identifying the cheater(s) and (ii) We assume that all servers always reply back (possibly with an incorrect answer).
   
\paragraph{Trivial Solution.}
We begin with the following direct generalization of Theorem~\ref{thm:main-single} to the multiple server case: essentially run $s$ independent copies of the protocol from Theorem~\ref{thm:main-single}. 

\begin{theorem}
\label{thm:main-multiple-trivial}
For every computable error-correcting code $H: [q]^{k/s} \rightarrow [q]^n$ that is $(\rho,L)$ list-decodable, there exists an $(s,s)$-party verification protocol with resource bound $(\log{n} + s\log{q},s(\log{n} + \log{q}))$ and verification guarantee $(1-\rho,f)$, where  for any $x\in [q]^k$, $f(x)=C(x) -s -\log (s^2qL^sn^4) -2\log\log(qn) -c_0$, for some fixed positive integer $c_0$.
\end{theorem}

\paragraph{Multiple Parties, One Hash.}
One somewhat unsatisfactory aspect of Theorem~\ref{thm:main-multiple-trivial} is that the storage needed by $U$ goes up a factor of $s$ from that in Theorem~\ref{thm:main-single}. Next we show that if the code $H$ is linear (and list decodable) then we can get a similar guarantee as that of Theorem~\ref{thm:main-multiple-trivial} except that the storage usage of $U$ remains the same as that in Theorem~\ref{thm:main-single}.

%

\begin{theorem}
\label{thm:main-multiple-linear}
For every computable linear error-correcting code $H: \F_q^k \rightarrow \F_q^n$ that is $(\rho,L)$ list-decodable, there exists an $(s,s)$-party verification protocol with resource bound $(\log{n} + \log{q},s(\log{n} + \log{q}))$ and verification guarantee $(1-\rho,f)$, where  for any $x\in \F_q^k$, $f(x)=C(x) -s -\log (s^2qLn^4) -2\log\log(qn) -c_0$, for some fixed positive integer $c_0$.
\end{theorem}

\paragraph{Catching the Cheaters and Handling Unresponsive Servers.}
%
%
%
The protocol in Theorem~\ref{thm:main-multiple-trivial} checks if the answer from each server is the same as corresponding stored hash. This implies that the protocol can easily handle the case when some server does not reply back at all. Additionally, if the protocol outputs a $0$ then it knows that at least one of the servers in the colluding set is cheating. (It does not necessarily identify the exact set $T$.\footnote{We assume that identifying at least one server in the colluding set is motivation enough for servers not to collude.})

However, the protocol in Theorem~\ref{thm:main-multiple-linear} cannot identify the cheater(s) and needs all the servers to always reply back. Next, using Reed-Solomon codes, at the cost of higher user storage and a stricter bound on the number of colluding servers, we show how to get rid of these shortcomings.

Recall that a Reed-Solomon code $RS:\F_q^m\rightarrow\F_q^{\l}$ can be represented as a systematic code (i.e. the first $k$ symbols in any codeword is exactly the corresponding message) and can correct $r$ errors and $e$ erasures as long as $2r+e\le \l-m$. Further, one can correct from $r$ errors and $e$ erasures in $O(\l^3)$ time. The main idea in the following result is to follow the protocol of Theorem~\ref{thm:main-multiple-trivial} but instead of storing all the $s$ hashes, $U$ only stores the parity symbols in the corresponding Reed-Solomon codeword.

\begin{theorem}
\label{thm:multiple-rs}
For every computable linear error-correcting code $H: \F_q^k \rightarrow \F_q^n$ that is $(\rho,L)$ list-decodable, assuming at most $e$ servers will not reply back to a challenge, there exists an $(r,s)$-party verification protocol with resource bound $(\log{n} +(2r+e)\cdot \log{q}),s(\log{n} + \log{q}))$ and verification guarantee $(1- \rho,f)$, where  for any $x\in \F_q^k$, $f(x)=C(x) -s -\log (s^2qLn^4) -2\log\log(qn) -c_0$, for some fixed positive integer $c_0$.

\end{theorem}

\section{Corollaries}

We now present specific instantiations of list decodable codes $H$ to obtain corollaries of our main results.

   \paragraph{Optimal Storage Enforcement.}
We begin with the following observation: If the reply from a prover comes from a domain of size $q$, then one cannot hope to have a $(\delta,f)$-\se\ hash family for any $\delta\le 1/q$ for any non-trivial $f$. This is because the prover can always return a random value and guess any $h_i(x)$ with probability $1/q$.

Next, we show that we can get $\delta$ to be arbitrarily close to $1/q$ while still obtaining $f(x)$ to be very close to $C(x)$. 
We start off with the following result due to Zyablov and Pinsker:
\begin{theorem}[\cite{ZP}]
Let $q\ge 2$ and let $0<\rho<1-1/q$. There exists a $(\rho,L)$-list decodable code with rate $1-H_q(\rho)-1/L$.
\end{theorem}

It is known that for $\eps<1/q$, $H_q(1-1/q-\eps)\le 1- C_q\eps^2$, where $C_q=q/(4\ln{q})$~\cite[Chap. 2]{a-thesis}. This implies that there exists a code $H:\F_q^k\rightarrow \F_q^n$, with $n\le \frac{k}{(C_q-1)\eps^2}\le 8k\ln{q}/(q\eps^2)$, which is $(1-1/q-\eps,1/\eps^2)$-list decodable. 
Note that the above implies that one can deterministically compute a uniquely-determined such code by iterating over all possible codes with dimension $k$ and block length $n$ and outputting the lexicographically least such one that is $(1-1/q-\eps,L)$-list decodable with the smallest discovered value of $L$.
Applying this to Theorem~\ref{thm:ld-se} implies the following optimal result:

\begin{corollary}
For every $\eps<1/q$ and integer $s\ge 1$,
there exists a  $(1/q+\eps,f)$-\se\ hash family, where for any $x\in[q]^k$, $f(x)=C(x) - s -\log{s^2k^4} +\log\eps^{2s+8} -\log\log{q^6k^2} +\log\log\eps^4 -\log\log\log{q^3} -c_0$ for some fixed positive integer $c_0$.
\end{corollary}

Some of our results need $H$ to be linear.
To this end, we will need the following result due to Guruswami et al.\footnote{The corresponding result for general codes has been known for more than thirty years.}

\begin{theorem}[\cite{GHK11}]
Let $q\ge 2$ be a prime power and let $0<\rho<1-1/q$. Then a random linear code of rate $1-H_q(\rho)-\eps$ is $(\rho,C_{\rho,q}/\eps)$-list decodable for some term $C_{\rho,q}$ that just depends on $\rho$ and $q$.
\end{theorem}

As a Corollary the above implies (along with the arguments used earlier in this section) that there exists a linear code $H:\F_q^k\rightarrow\F_q^n$ with $n\le 8k\ln{q}/(q\eps^2)$ that is $(1-1/q-\eps,C'_{\eps,q}/\eps^2)$-list decodable (where $C'_{\eps,q}\stackrel{def}=C_{1-1/q-\eps,q}$). Applying this to Theorem~\ref{thm:multiple-rs} gives us the following

\begin{corollary}
For every $\eps<1/q$, integer $s\ge 1$, and $r,e\le s$, assuming at most $e$ servers do not reply back to a challenge,
there exists an $(r,s)$-party verification protocol with resource bound $(\log{kq^{2r+e-1}}-\log{\eps^2}+\log\log{q^{3/2}+3},s(\log{k}-\log{\eps^2}+\log\log{q^{3/2}}+3))$ and verification guarantee $(1/q+\eps,f)$, where for any $x\in[q]^k$, $f(x)=C(x) -s -\log{s^2C'_{\eps,q}k^4}+\log{q^3\eps^{10}}-\log\log{q^6k^2}+\log\log{\eps^2}-\log\log\log{q^2}-c_0$ for some fixed positive integer $c_0$.

\end{corollary}

\paragraph{Proof of Theorem~\ref{thm:au-se}.}
Corollary~\ref{cor:au-ld} and Theorem~\ref{thm:ld-se} proves Theorem~\ref{thm:au-se} with $f(x)=C(x)-O(\log{q}+\log{n})$.

\paragraph{Practical Storage Enforcement.}

All of our results so far have used (merely) computable codes $H$, which are not that useful in practice. What we really want in practice is to use codes $H$ that lead to an efficient implementation of the protocol. At the very least, all the honest parties in the verification protocol should not have to use more than polynomial time to perform the required computation. An even more desirable property would be for honest parties to be able to do their computation in a one pass, logspace, data stream fashion. In this section, we'll see one example of each. Further, it turns out that the resulting hash functions are classical ones that are also used in practice. In particular, using the Karp-Rabin hash and the polynomial hash we get the following results. (More details are in the appendix.)

\begin{corollary}
For every $\eps>0$, there exists an $(s,s)$-party verification protocol with resource bound $$\left(
(s+1)\log{k} + s - \log{\eps^4} + s\log\log({k/\eps^2}), s\log{k}-s\log{\eps^2} + s + s\log\log{(k/\eps^2)}\right)$$
with verification guarantee $(\eps,f)$, where for every 
$x\in\{0,1,\dots,\prod_{i=1}^k p_i-1\}$,
$$f(x)=C(x) -c_0\left(s(\log(k/\eps)-\log\log(k/\eps)-1)\right) -c_1\log\log\log(k/\eps) -c_2$$
for some fixed positive integers $c_0,c_1$ and $c_2$. Further, all honest parties can do their computation in $\mathrm{poly}(n)$ time.
\end{corollary}

\begin{corollary}
For every $\eps>0$,
\begin{itemize}
\item[(i)] There exists an $(s,s)$-party verification protocol with resource bound $( (s+1)(\log{k}+2\log(1/\eps)+1), 2s(\log{k}+2\log(1/\eps)+1))$ and verification guarantee $(\eps,f)$, where for any $x\in \F_q^k$, $f(x)=C(x)-O(s(\log{k}+\log(1/\eps)))$.
\item[(ii)] Assuming at most $e$ servers do not respond to challenges, there exists an $(r,s)$-party verification protocol with resource bound $( (2r+e+1)(\log{k}+2\log(1/\eps)+1), 2s(\log{k}+2\log(1/\eps)+1))$ and verification guarantee $(\eps,f)$, where for any $x\in \F_q^k$, $f(x)=C(x)-O(s+\log{k}+\log(1/\eps))$.
\end{itemize}
\noindent
Further, in both the protocols, honest parties can implement their required computation with a one pass, $O(\log{k}+\log(1/\eps))$ space (in bits) and $\tilde{O}(\log{k}+\log(1/\eps))$ update time data stream algorithm.
\end{corollary}



\section*{Acknowledgments} We thank Dick Lipton for pointing out the application of our protocol to the OS updating problem (a special case of proof of secure erasure) and for kindly allowing us to use his observation. We also thank Ram Sridhar for helpful discussions.

\bibliographystyle{abbrv}
\bibliography{verify}

\appendix

%
\section{Related Works}
Existing approaches for data possession verification at remote storage can be broadly classified into two categories: Crypto-based and Coding based. \textit{Crypto-based approaches} rely on symmetric and asymmetric cryptographic primitives for proof of data possession. Ateniese et al. \cite{AteniesePDP} defined the proof of data possession (PDP) model which uses public key homomorphic tags for verification of stored files. It can also support public verifiability with a slight modification of the original protocol by adding extra communication cost. In subsequent work, Ateniese et al. \cite{AtenieseSEPDP} proposed a symmetric crypto-based variation (SEP) which is computationally efficient compared to the original PDP but lacks public verifiability. Also, both of these protocols considered the scenario with files stored on a single server, and do not discuss erasure tolerance. However, Curtmola et al. \cite{CurtmolaMRPDP} extended PDP to a multiple-server scenario by introducing multiple identical replicas of the original data. Among other notable constructions of PDP, Gazzoni et al.\cite{GazzoniDDP} proposed a scheme (DDP) that relied on an RSA-based hash (exponentiating the whole file), and Shah et al. \cite{ShahPPAUDIT} proposed a symmetric encryption based storage audit protocol. Recent extensions of crypto-based PDP schemes by Wang et al. (EPV) \cite{WangEPVDD} and Erway et al. \cite{ErwayDPDP} mainly focus on supporting data dynamics in addition to existing capabilities. Golle et al. \cite{GolleSEC} had proposed a cryptographic primitive called storage enforcing commitment (SEC) which probabilistically guarantees that the server is using storage whose size is equal to the size of the original data to correctly answer the data possession queries. In general, the drawbacks of the aforementioned protocols are: \textit{(a)} being computation intensive due to the usage of expensive cryptographic primitives and \textit{(b)} since each verification checks a random fragment of the data, a small fraction of data corruption might go undetected and hence they do not guarantee the retrievability of the original data. \textit{Coding-based approaches}, on the other hand, have relied on special properties of linear codes such as the Reed-Solomon (RS) \cite{ReedRS} code. The key insight is that encoding the data imposes certain algebraic constraints on it which can be used to devise an efficient fingerprinting scheme for data verification. Earlier schemes proposed by Schwarz et al. (SFC) \cite{SchwarzSFC} and Goodson et al. \cite{GoodsonEBT} are based on this and are primarily focused on the construction of fingerprinting functions and categorically fall under distributed protocols for file integrity checking. Later, Juels and Kaliski \cite{JuelsPOR} proposed a construction of a proof of retrievability (POR) which guarantees that if the server passes the verification of data possession, the original data is retrievable with high probability. While the scheme by Juels \cite{JuelsPOR} supported a limited number of verifications, the theoretical POR construction by Shacham and Waters \cite{ShachamCPOR} extended it to unlimited verification and public verifiability by integrating cryptographic primitives. Subsequently, Dodis et al. \cite{DVW09} provided theoretical studies on different variants of existing POR schemes and Bowers et al. \cite{BowersPOR} considered POR protocols of practical interest \cite{JuelsPOR, ShachamCPOR} and showed how to tune parameters to achieve different performance goals. However, these POR schemes only consider the single server scenario and have no construction of a retrievability and storage enforcement guarantee in a distributed storage scenario.

\section{Omitted Proofs}

\subsection{Proof of Proposition~\ref{prop:au-verification}}
\begin{proof}
We begin by specifying the protocol. In the preprocessing step, the verifier $U$ does the following on input $x\in [q]^k$:
\begin{enumerate}
\item Generate a random $\beta\in [n]$.
\item Store $(\beta,\alpha=h_{\beta}(x))$ (and in the data possession problem, send $x$ to the prover).
\end{enumerate}
The prover, upon receiving $x$, saves a string $y\in [q]^*$. The prover is allowed to use any computable function to obtain $y$ from $x$.

During the verification phase, $U$ does the following:
\begin{enumerate}
\item It sends $\beta$ to the prover.\item It receives $a\in [q]$ from the prover. ($a$ is supposed to be $h_{\beta}(x)$.)
\item It outputs $1$ (i.e. prover did not ``cheat") if $a=\alpha$, else it outputs a $0$.
\end{enumerate}
The resource bound follows from the specification of the protocol and the performance guarantee follows from the fact that $\H$ is $(\gamma,
f)$-\se.
\end{proof}

\subsection{Proof of Proposition~\ref{prop:au-code}}
\begin{proof} This proof is straightforward. For $x\neq y\in [q]^k$, since $\H$ is $\eps$-\au, there exists at most $\eps n$ values of $i\in [n]$ such that $h_i(x)=h_i(y)$, which implies that $\Delta(H(x),H(y))\ge (1-\eps)n$, which proves the claim.
\end{proof}

\subsection{Proof of Theorem~\ref{thm:main-multiple-trivial}}
\begin{proof}We begin by specifying the protocol. In the pre-processing step, the client $U$ does the following
on input $x\in [q]^k$:

\begin{enumerate}
\item Generate a random $\beta\in [n]$.
\item Store $(\beta,\gamma_1=H(x_1)_{\beta},\dots,\gamma_s=H(x_s)_{\beta})$ and send $x_i$ to the server $i$ for
every $i\in [s]$.
\end{enumerate}

Server $i$ on receiving $x$, saves a string $y_i\in [q]^*$. The server is allowed to use any computable function
to obtain $y_i$ from $x_i$.

During the verification phase, $U$ does the following:

\begin{enumerate}
\item It sends $\beta$ to all $s$ servers.
\item It receives $a_i\in [q]$ from  server $i$ for every $i\in [s]$. ($a_i$ is supposed to be $H(x_i)_{\beta}$.)
\item It outputs $1$ (i.e. none of the servers ``cheated'') if $a_i=\gamma_i$ for every $i\in [s]$, else it outputs a $0$.
\end{enumerate}

Similar to the one-party result, we assume that server $i$, on receiving the challenge, uses a computable function
$\A_{x,i}:[n]\times [q]^*\rightarrow [q]$ to compute $a_i=\A_x(\beta,y_i)$ and sends $a_i$ back to $U$.

The claim on the resource usage follows immediately from the protocol specification. Next we prove its verification
guarantee. Let $T\subseteq [s]$ be the set of colluding servers. We will prove that $y_T$ is large by contradiction:
if not, then using the list decodability of $H$, we will present a description of $x_T$ of size $<C(x_T)$.
Consider the following algorithm that uses $y_T$ and an advice string $v\in\left(\{0,1\}^{|L|}\right)^{|T|}$, which
is the concatenation of shorter strings $v_{i} \in \left(\{0,1\}^{|L|}\right)$ for each $i\in T$:

\begin{enumerate}
\item Compute a description of $H$ from $n,k,\rho n,q$ and $s$.
\item For every $j\in T$, compute $z_j=(\A_{x,j}(\beta,y_j))_{\beta\in [n]}$.
\item Do the following for every $j\in T$: by cycling through all $x_j\in [q]^{k/s}$, retain the
   set $\mathcal{L}_j\subseteq [q]^{k/s}$ such that for every $u\in \mathcal{L}_j$, $\Delta(H(u),z_j) \le \rho n$.
\item For each $j \in T$, let $w_j$ be the $v_j$th string from $\mathcal{L}_j$.
\item Output the concatenation of $\{w_j\}_{j\in T}$.
\end{enumerate}

Note that since $H$ is $(\rho,L)$-list decodable, there exists an advice string $v$ such that the algorithm above
outputs $x_T$. Further, since $H$ is computable, there is an algorithm $\mathcal E$ that can compute a description
of $H$ from $n,k\rho n, q$ and $s$. (Note that using this description, we can generate any codeword $H(u)$ in step 3.)
Thus, we have description of $x_T$ of size
$|y_T|+|v|+\sum_{j\in T}|\A_{x,j}|+|\mathcal E|+ (s +\log (s^2qL^sn^4) +2\log\log(qn) +3)$
(where the term in parentheses is for encoding the different parameters and $T$), which means that
if $|y_T| < C(x_T)-|v|-\sum_{j\in T}|\A_{x,j}|-|\mathcal E| - (s +\log (s^2qL^sn^4) +2\log\log(qn) +3) = f(x)$,
then we have a description of $x_T$ of size $<C(x_T)$, which is a contradiction.
\end{proof}

\subsection{Proof of Theorem~\ref{thm:main-multiple-linear}}

\begin{proof}We begin by specifying the protocol. In the pre-processing step, the client $U$ does the following on input $x\in [q]^k$:

\begin{enumerate}
\item Generate a random $\beta\in [n]$.
\item Store $(\beta,\gamma=H(x)_{\beta})$ and send $x_i$ to the server $i$ for every $i\in [s]$.
\end{enumerate}

Server $i$ on receiving $x_i$, saves a string $y_i\in [q]^*$. The server is allowed to use any computable function to obtain $y_i$ from $x_i$.
For notational convenience, we will use $\hat{x}_i$ to denote the string $x_i$ extended to a string in $\F_q^k$ by adding zeros in positions that correspond to servers other than $i$.

During the verification phase, $U$ does the following:

\begin{enumerate}
\item It sends $\beta$ to all $s$ servers.
\item It receives $a_i\in [q]$ from  server $i$ for every $i\in [s]$. ($a_i$ is supposed to be $H(\hat{x}_i)_{\beta}$.)
\item It outputs $1$ (i.e. none of the servers ``cheated") if $\gamma=\sum_{i=1}^s a_i$ else it outputs a $0$.
\end{enumerate}

We assume that server $i$ on receiving the challenge, uses a computable function $\A_{x,i}:[n]\times [q]^*\rightarrow [q]$ to compute $a_i=\A_x(\beta,y_i)$ and sends $a_i$ back to $U$.
The claim on the resource usage follows immediately from the protocol specification. Next we prove its verification guarantee. Let $T\subseteq [s]$ be the set of colluding servers. We will prove that $y_
T$ is large by contradiction: if not, then using the list decodability of $H$, we will present a description of $x_T$ of size $<C(x_T)$.

For notational convenience, define $\hat{x}_T=\sum_{j\in T} \hat{x}_j$ and $\hat{x}_{\overline{T}}=\sum_{j\not\in T} \hat{x}_j$.
Consider the following algorithm that uses $y_T$ and an advice string $v\in\{0,1\}^{|L|}$:

\begin{enumerate}
\item Compute a description of $H$ from $n,k,\rho,q,s$ and $L$.
\item Compute $z=(\sum_{j\in T} \A_{x,j}(\beta,y_j))_{\beta\in [n]}$.
\item By cycling through all $x\in \F_q^k$, retain the set $\mathcal{L}\subseteq \F_q^k$ such that for every $u\in \mathcal{L}$, $\Delta(H(u),z) \le \rho n$.
\item Output the $v$th string from $\mathcal{L}$.
\end{enumerate}

To see the correctness of the algorithm above, note that for every $j\in [s]\setminus T$, $(\A_{x,j}(\beta,y_j))_{\beta\in [n]}=H(\hat{x}_j)$. Thus, if the protocol outputs $1$ with probability at least $1-\rho$, then $\delta(z,H(\hat{x}_T))\le \rho n$ ; here we used the linearity of $H$ to note that $H(\hat{x}_T)=H(x)-H(\hat{x}_{\overline{T}})$. Note that since $H$ is $(\rho,L)$-list decodable, there exists an advice string $v$ such that the algorithm above outputs $\hat{x}_T$ (from which we can easily compute $x_T$). Further, since $H$ is computable, there is an algorithm $\mathcal E$ that can compute a description of $H$ from $s,n,k,\rho n$ and $q$.
Thus, we have a description of $x_T$ of size $|y_T|+|v|+\sum_{j\in T}|\A_{x,j}|+|\mathcal E|+ (s +\log (s^2qLn^4) +2\log\log(qn) +3)$, (where the term in parentheses is for encoding the different parameters and $T$), which means that if $|y_T| < C(x_T)-|v|-|\A_x|-|\mathcal E|-(s +\log (s^2qLn^4) +2\log\log(qn) +3) = f(x)$, then we have a description of $x_T$ of size $<C(x_T)$, which is a contradiction.
\end{proof}

\subsection{Proof of Theorem~\ref{thm:multiple-rs}}
\begin{proof}
We begin by specifying the protocol. As in the proof of
Theorem~\ref{thm:main-multiple-linear}, define $\hat{x}_i$, for $i\in [s]$,
to be the string $x_i$ extended to the vector in $\F_q^k$, which has zeros
in the positions that do not belong to server $i$. Further, for any subset
$T\subseteq [s]$, define $\hat{x}_T=\sum_{i\in T}\hat{x}_i$. Finally let
$RS:\F_q^s\rightarrow\F_q{\l}$ be a systematic Reed-Solomon code where
$\l=2r+e+s$.

In the pre-processing step, the client $U$ does the following on input
$x\in [q]^k$:

\begin{enumerate}
\item Generate a random $\beta\in [n]$.
\item Compute the vector$v=(H(\hat{x}_1)_{\beta},\dots,H(\hat{x}_s)_{\beta})\in\F_q^s$.\item Store $(\beta,\gamma_1=RS(v)_{s+1},\dots,\gamma_{2r+e}=RS(v_{\l}))$
and send $x_i$ to the server $i$ for every $i\in [s]$.
\end{enumerate}

Server $i$ on receiving $x_i$, saves a string $y_i\in [q]^*$. The server is
allowed to use any computable function to obtain $y_i$ from $x_i$.


During the verification phase, $U$ does the following:

\begin{enumerate}
\item It sends $\beta$ to all $s$ servers.
\item For each server $i\in [s]$, it either receives no response or
receives $a_i\in \F_q$. ($a_i$ is supposed to be $H(\hat{x}_i)_{\beta}$.)
\item It computes the received word $z\in\F_q^{\l}$, where for $i\in [s]$,
$z_i=?$ (i.e. an erasure) if the $i$th server does not respond else
$z_i=a_i$ and for $s<i\le \l$, $z_i=\gamma_i$.
\item Run the decoding algorithm for $RS$ to compute the set
$T'\subseteq [s]$ to be the error locations. (Note that by Step 2,
$U$ already knows the set $E$ of erasures.)
\end{enumerate}

We assume that server $i$ on receiving the challenge, uses a computable
function $\A_{x,i}:[n]\times [q]^*\rightarrow [q]$ to compute
$a_i=\A_x(\beta,y_i)$ and sends $a_i$ back to $U$ (unless it decides not
to respond).

The claim on the resource usage follows immediately from the protocol specification. We now prove the verification guarantee. Let $T$ be the
set of colluding servers. We will prove that with probability at least
$1-\rho$, $U$ using the protocol above computes
$\emptyset\neq T'\subseteq T$ (and $|y_T|$ is large enough). Fix a
$\beta\in [n]$. If for this $\beta$, $U$ obtains $T'=\emptyset$, then
this implies that for every $i\in [s]$ such that server $i$ responds, we
have $a_i=H(\hat{x}_i)_{\beta}$. This is because of our choice of $RS$,
the decoding in Step 4 will return $v$ (which in turn allows us to compute
exactly the set $T'\subseteq T$ such that for every
$j\in T'$, $a_j\neq H(\hat{x})_{\beta}$).\footnote{We will assume that
$T\cap E=\emptyset$. If not, just replace $T$ by $T\setminus E$.} Thus,
if the protocol outputs a $T'\neq\emptyset$ with probability at least
$1-\rho$ over the random choices of $\beta$, then using the same argument as
in the proof of Theorem~\ref{thm:main-multiple-linear}, we note that
$\Delta(H(\hat{x}_T),(\sum_{j\in T}\A_{x,j}(\beta,y_j))_{\beta\in [n]})\le \rho n$.
Again, using the same argument as in the proof of Theorem~\ref{thm:main-multiple-linear} this
implies that $|y_T|\ge C(x_T) -s -\log (s^2qLn^4) -2\log\log(qn) -c_0$, for some fixed positive
integer $c_0$.
\end{proof}


   \paragraph{Hashing Modulo a Random Prime.}
We will begin with a code that corresponds to the classical Karp-Rabin hash~\cite{KarpRabin}.  It is known that the corresponding code $H$ is the so called CRT code. Let $H$ be the so called Chinese Remainder Theorem (or CRT) codes. In particular, we will consider the following special case of such codes. Let $p_1\le p_2\le \cdots\le p_n$ be the first $n$ primes. Consider the CRT code $H:\prod_{i=1}^k [p_i]\rightarrow \prod_{i=1}^n [p_i]$, where the message $x\in \{0,1,\dots,(\prod_{i=1}^k p_i)-1\}$, is mapped to the vector $( x\mod p_1,x\mod p_2, \dots, x\mod p_n)\in \prod_{i=1}^n [p_i]$. It is known that such codes have distance $n-k+1$ (cf.~\cite{G-thesis}). By a simple upper bound on the prime counting function (cf.~\cite{bach-shallit}), we can take $p_n \le 2 n\log{n}$. 
Moreover, $\sum_{i=1}^n p_i < {np_n}/2$ (cf.~\cite{Rosser-Schoenfeld}). Thus, if we pick a CRT code with $n= k/\eps^2$, then by Theorem~\ref{thm:johnson}, $H$ is $(1-\eps, k^2(\log k - \log{\eps^2})/\eps^4)$
-list decodable.
Given any $x\in \{0,1,\dots,(\prod_{i=1}^k p_i)-1\}$ and a random $\beta\in [n]$, $H(x)_{\beta}$ corresponds to the Karp-Rabin fingerprint (modding the input integer with a random prime). Further, $H(x)_{\beta}$ can be computed in polynomial time.
Thus, letting $H$ be the CRT code in Theorem~\ref{thm:main-multiple-trivial}, we get the following:
\begin{corollary}
For every $\eps>0$, there exists an $(s,s)$-party verification protocol with resource bound $$\left(
(s+1)\log{k} + s - \log{\eps^4} + s\log\log({k/\eps^2}), s\log{k}-s\log{\eps^2} + s + s\log\log{(k/\eps^2)}\right)$$
with verification guarantee $(\eps,f)$, where for every
$x\in\{0,1,\dots,\prod_{i=1}^k p_i-1\}$,
$$f(x)=C(x) -c_0\left(s(\log(k/\eps)-\log\log(k/\eps)-1)\right) -c_1\log\log\log(k/\eps) -c_2$$
for some fixed positive integers $c_0,c_1$ and $c_2$. Further, all honest parties can do their computation in $\mathrm{poly}(n)$ time.
\end{corollary}
\begin{remark}
Theorem~\ref{thm:multiple-rs} can be extended to handle the case where the symbols in codewords of $H$ are of different sizes. However, for the sake of clarity we refrain from applying CRT to the generalizat
ion of Theorem~\ref{thm:multiple-rs}. Further, the results in the next subsection allow for a more efficient implementation of the computation required from the honest parties.
\end{remark}

\paragraph{Reed-Solomon Codes.} Finally, we take $H:\F_q^k\rightarrow \F_q^n$ to be the  Reed-Solomon code, with $n=q$. Recall that for such a code, given message $x=(x_0,\dots,x_{k-1})\in\F_q^k$, the codeword is given by $H(x)=(P_x(\beta))_{\beta\in\F_q}$, where $P_x(Y)=\sum_{i=0}^{k-1} x_iY^i$. It is well-known that such a code $H$ has distance $n-k+1$. Thus, if we pick $n=k/\eps^2$, then by Theorem~\ref{thm:johnson}, $H$ is $(1-\eps,2k^2/\eps^4)$-list decodable.Let $H$ be the Reed-Solomon code (more details in the appendix). Given any $x\in \F_q^k$ and a random $\beta\in [n]$, $H(x)_{\beta}$ corresponds to the widely used ``polynomial" hash. Further, $H(x)_{\beta}$
 can be computed in one pass over $x$ with storage of only a constant number of $\F_q$ elements. (Further, after reading each entry in $x$, by the Homer's rule, the algorithm just needs to perform one additi
on and one multiplication over $\F_q$.)
Thus, applying $H$ as the Reed-Solomon code to Theorems~\ref{thm:main-multiple-trivial} and~\ref{thm:multiple-rs} implies the following:
\begin{corollary}
For every $\eps>0$,\begin{itemize}
\item[(i)] There exists an $(s,s)$-party verification protocol with resource bound $( (s+1)(\log{k}+2\log(1/\eps)+1), 2s(\log{k}+2\log(1/\eps)+1))$ and verification guarantee $(\eps,f)$, where for any $x\in \F_q^k$, $f(x)=C(x)-O(s(\log{k}+\log(1/\eps)))$.
\item[(ii)] Assuming at most $e$ servers do not respond to challenges, there exists an $(r,s)$-party verification protocol with resource bound $( (2r+e+1)(\log{k}+2\log(1/\eps)+1), 2s(\log{k}+2\log(1/\eps)+1
))$ and verification guarantee $(\eps,f)$, where for any $x\in \F_q^k$, $f(x)=C(x)-O(s+\log{k}+\log(1/\eps))$.
\end{itemize}\noindent
Further, in both the protocols, honest parties can implement their required computation with a one pass, $O(\log{k}+\log(1/\eps))$ space (in bits) and $\tilde{O}(\log{k}+\log(1/\eps))$ update time data strea
m algorithm.
\end{corollary}

\end{document}